\documentclass[pdflatex,sn-mathphys-num]{sn-jnl}%

\usepackage{graphicx}%
\usepackage{multirow}%
\usepackage{amsmath,amssymb,amsfonts}%
\usepackage{amsthm}%
\usepackage{mathrsfs}%
\usepackage[title]{appendix}%
\usepackage{xcolor}%
\usepackage{textcomp}%
\usepackage{manyfoot}%
\usepackage{booktabs}%
\usepackage{algorithm}%
\usepackage{algorithmicx}%
\usepackage{algpseudocode}%
\usepackage{listings}%
\usepackage{subcaption}
\usepackage{color}
\usepackage{threeparttable}
\usepackage{makecell}
\usepackage{hyperref}

\theoremstyle{thmstyleone}%
\newtheorem{theorem}{Theorem}%

\theoremstyle{thmstylethree}%
\newtheorem{definition}{Definition}
\newtheorem{lemma}{Lemma}

\begin{document}

\title[Article Title]{Optimistic $\epsilon$-Greedy Exploration for Cooperative Multi-Agent Reinforcement Learning}

\author[1]{\fnm{Ruoning} \sur{Zhang}}\email{zhangruoning@std.uestc.edu.cn}

\author*[2]{\fnm{Siying} \sur{Wang}}\email{siyingwang@swun.edu.cn}

\author[1]{\fnm{Wenyu} \sur{Chen}}\email{cwy@uestc.edu.cn}

\author[1]{\fnm{Yang} \sur{Zhou}}\email{zhouy@std.uestc.edu.cn}

\author[3]{\fnm{Zhitong} \sur{Zhao}}\email{zzt@cdut.edu.cn}

\author[1]{\fnm{Zixuan} \sur{Zhang}}\email{zhangzixuan@std.uestc.edu.cn}

\author[1]{\fnm{Ruijie} \sur{Zhang}}\email{ruijie\_zhang@std.uestc.edu.cn}

\author[4]{\fnm{Stefano V.} \sur{Albrecht}}\email{stefano.albrecht@ntu.edu.sg}

\affil*[1]{\orgdiv{School of Computer Science and Engineering}, \orgname{University of Electronic Science and Technology of China}, \orgaddress{\city{Chengdu}, \country{China}}}

\affil[2]{\orgdiv{College of Computer Science and Artificial Intelligence}, \orgname{Southwest Minzu University}, \orgaddress{\city{Chengdu},  \country{China}}}

\affil[3]{\orgdiv{College of Nuclear Technology and Automation Engineering}, \orgname{Chengdu University of Technology}, \orgaddress{\city{Chengdu}, \country{China}}}

\affil[4]{\orgdiv{College of Computing and Data Science}, \orgname{Nanyang Technological University}, \orgaddress{\country{Singapore}}}

\abstract{The Centralized Training with Decentralized Execution (CTDE) paradigm is widely used in cooperative multi-agent reinforcement learning. However, conventional methods based on CTDE can suffer from value underestimation and converge to suboptimal solutions. While such underestimation is typically attributed to the representational limitations of monotonic structures, we provide a novel perspective by demonstrating that the insufficient sampling of optimal joint actions during exploration is also a critical factor. To address this problem, we propose Optimistic $\epsilon$-Greedy Exploration. Our method introduces optimistic action-value networks that serve as decoupled exploration indicators, which we theoretically prove to converge in probability to the maximum achievable returns. By sampling actions from these distributions with a probability of $\epsilon$, we effectively increase the selection frequency of high-return joint actions. Experimental results in various environments reveal that our strategy effectively prevents the algorithm from falling into suboptimal solutions and significantly improves final returns, win rates, and convergence speeds compared to other enhanced algorithms. Our code has been open-sourced at \href{https://github.com/qxqxtxdy/OptimisticExploration}{OptimisticExploration}.}

\keywords{Multi-agent reinforcement learning, Centralized training with decentralized execution, Value decomposition underestimation, Optimistic exploration}

\maketitle

\section{Introduction}\label{sec1}

Reinforcement learning (RL) is a paradigm for modeling the interaction between agent and environment, attracting extensive attention in recent years. Under this paradigm, an agent interacts with the environment through trial and error, optimizing its policy to maximize cumulative rewards. Powered by the development of deep neural networks, RL has demonstrated remarkable potential in complex tasks such as autonomous driving \cite{wang2021decision}, computer games \cite{liu2025aligngap}, and recommendation systems \cite{gao2023cirs}.

However, conventional single-agent RL algorithms encounter significant limitations in real-world scenarios such as multi-robot control \cite{lowe2017multi} and multi-player computer games \cite{shmakov2019colosseumrl}, as these situations require algorithms to consider the interaction between multiple agents. This critical gap has driven the rapid development of Multi-Agent Reinforcement Learning (MARL) \cite{albrecht2024multi}, which extends RL algorithms to multi-agent environments. In fully cooperative MARL, the objective is to maximize a shared cumulative reward through teamwork, requiring agents to strategically coordinate their actions. To achieve this, Centralized Training with Decentralized Execution (CTDE) has emerged as a widely adopted paradigm \cite{amato2024introduction}. CTDE enables agents to learn cooperative policies by leveraging global information during centralized training, while enabling decentralized execution after training by conditioning the policies only on the agents' local observations.

Among the various approaches in CTDE, Value Decomposition (VD) simplifies complex multi-agent coordination by factorizing a global action-value function into individual local components \cite{sunehag2017value}. To ensure consistency between global function and individual components, the joint action maximizing the global estimation is required to match the combination of each agent's locally optimal actions \cite{son2019qtran}, known as the Individual Global Max (IGM) condition. Conventional VD algorithms typically satisfy this constraint by aligning local and global estimations via monotonic aggregation functions. For instance, VDN \cite{sunehag2017value} uses the sum of local estimations as the global estimation, while QMIX \cite{rashid2020monotonic} employs a monotonic hypernetwork with non-negative parameters to meet this requirement. 

Although these algorithms have shown promising performance in various tasks \cite{papoudakis2021benchmarking,hu2021rethinking}, a potential drawback lies in the limited accuracy of their value estimations. A typical example involves scenarios where the optimal joint actions are surrounded by detrimental ones, such that any agent deviating from the optimal joint action can cause a steep decline in the team reward. Table \ref{addtable:1}(a) exemplifies this problem within a matrix game setting. While the maximum reward of 8 is attained only when all agents jointly execute action $a_1$, a unilateral deviation by any single agent severely penalizes the team, dropping the reward to -12. Table \ref{addtable:1}(b) details the estimation results produced by the conventional monotonic value decomposition algorithm QMIX. Because this estimation process incorporates the lower rewards yielded when other agents deviate from $a_1$, the value of the true optimal joint action is systematically underestimated, which ultimately traps the policy in suboptimal solution. Several studies have attempted to tackle this problem by introducing more complex network architectures involving credit assignment \cite{zhao2024qdap}, mutual information \cite{iqbal2019actor}, or consistency protocols \cite{zhang2018fully}. Nevertheless, the reliance on complex neural networks can introduce new challenges, such as training instability and slower convergence rates \cite{lyu2019gradient}. In highly intricate multi-agent environments, excessively complex architectures can significantly increase the probability of training failure \cite{gronauer2022multi}.

\begin{table*}[t]
    \centering
    \setlength{\tabcolsep}{5.5pt}
    \begin{subtable}[c]{0.4\textwidth}
        \centering
            \begin{tabular}{|c|c|c|c|}
				\hline
				    & $a_1$ & $a_2$ & $a_3$ \\
				\hline
				$a_1$ & \textcolor{blue}{8} & -12 & -12 \\
				\hline
				$a_2$ & -12 & 0 & 0 \\
				\hline
				$a_3$ & -12 & 0 & 0 \\
				\hline
            \end{tabular}
            \caption{Payoff Matrix}
    \end{subtable}
    \begin{subtable}[c]{0.4\linewidth}
        \centering
            \begin{tabular}{|c|c|c|c|}
				\hline
				  & $a_1$ & $a_2$ & $a_3$ \\
				\hline
				$a_1$ & -7.76 & -7.76 & -7.76 \\
				\hline
				$a_2$ & -7.76 & \textcolor{red}{0.04} & 0.02 \\
				\hline
				$a_3$ & -7.76 & 0.02 & 0.01 \\
				\hline
            \end{tabular}
        \caption{Estimation Results of QMIX}
    \end{subtable}
\caption{Instance of Systematic Value Underestimation within Matrix Game\label{addtable:1}}
\end{table*}

In this paper, we conduct an in-depth investigation into the mathematical principles underlying value underestimation to develop a more effective solution. As pointed out by ParetoAC \cite{christianos2023pareto}, value underestimation fundamentally stems from the agents' failure to robustly identify Pareto-optimal joint actions. Thus, our analysis addresses the critical question: what causes VD-based methods to deviate from optimal joint actions? Through the derivation of the neural networks' parameter update processes in VD-based methods, we uncover that actions are implicitly weighted according to their sampling frequency during the computation of the temporal-difference (TD) loss. Such weighted updates induce systematic shifts in value estimations towards actions with higher weights. If the actions receiving higher weights are not the actual optimal choices, the values of optimal joint actions may be underestimated. Based on our analysis, we propose a novel exploration strategy: \textbf{Optimistic $\epsilon$-Greedy Exploration}. This strategy maintains monotonically non-decreasing optimistic action-value networks that are trained to estimate the maximum expected return of each agent, converging in probability to identify optimal joint actions. During training, with a probability of $\epsilon$, the agents sample actions from these optimistic action-value networks to encourage more frequent selection of high-return actions. This strategy provides a theoretically-founded solution to the implicit suboptimal weighting problem that affects other exploration strategies, enabling agents to learn optimal joint actions. Experimental results across Matrix Games, Predator and Prey, and the StarCraft Multi-Agent Challenge (SMAC) environments demonstrate that our method effectively prevents agents from falling into suboptimal action selection, achieving higher final returns, improved win rates, and faster convergence rates. 

In summary, our main contributions are:

\begin{itemize}
\item We theoretically analyze the systematic underestimation of optimal joint actions in cooperative MARL, showing that the insufficient sampling of optimal joint actions significantly contributes to this underestimation.
\item Based on our analysis, we propose Optimistic $\epsilon$-Greedy Exploration, which utilizes monotonically non-decreasing optimistic action-value networks to estimate maximum expected returns. This strategy robustly guides policies toward optimal joint actions by increasing the sampling frequency of high-return joint actions.
\item We integrate our proposed exploration strategy within the VDN \cite{sunehag2017value} and QMIX \cite{rashid2020monotonic} frameworks and evaluate it across several cooperative multi-agent environments of varying complexity. The results demonstrate that our variants effectively resolves underestimation, achieving higher final returns, superior win rates, and faster convergence compared to the conventional VDN and QMIX algorithms as well as other enhanced value decomposition algorithms relying on manual weighting or complex architectures\cite{rashid2020weighted,son2019qtran}.
\end{itemize}

The remainder of this paper is organized as follows: In Section \ref{sec:2}, we review the related work on MARL. In Section \ref{sec:3}, we introduce the technical preliminaries of our approach. In Section \ref{sec:4}, we investigate the mathematical foundations behind the value underestimation problem, analyze the convergence properties of the optimistic action-value networks, and present the Optimistic $\epsilon$-Greedy Exploration method proposed based on these principles. In Section \ref{sec:5}, we demonstrate the comparative results across various cooperative multi-agent experimental environments and parameter settings. In Section \ref{sec:6}, we conclude our work and discuss potential future research directions.

\section{Related Work}\label{sec:2}

\textbf{Value Decomposition and Architectural Enhancements:} Value decomposition serves as a specific instantiation of the Centralized Training with Decentralized Execution (CTDE) paradigm within value-based multi-agent reinforcement learning. It addresses the trade-off between the inability of independent Q-learning \cite{littman1994markov} to learn cooperative policies and the inability of centralized Q-learning \cite{foerster2017stabilising,foerster2018counterfactual} to learn policies that can run independently for each agent under partial observability \cite{albrecht2024multi}. Representative value decomposition algorithms, such as VDN \cite{sunehag2017value} and QMIX \cite{rashid2020monotonic}, allow each agent to decide its action independently during execution, while developing cooperative policies during the training phase via a Mixer utilizing either summation or a hypernetwork architecture. Although such methods have demonstrated promising performance in tasks such as StarCraft II \cite{samvelyan2019starcraft}, the traditional $\epsilon$-greedy exploration strategy adopted by these methods introduces a potential risk of value underestimation. Some methods attempt to mitigate the shortcomings of the exploration strategy by introducing complex network architectures. QPLEX \cite{wang2020qplex} decouples local estimation into state value and action advantage to jointly estimate global value; QTRAN \cite{son2019qtran} introduces additional estimation terms to correct biases in value representation. ResQ \cite{shen2022resq} and RQN \cite{pina2022residual} employ residual-like networks to provide correction terms for each agent, mitigating the impacts of irregular reward distributions. Unfortunately, increasing the complexity of simple network introduces challenges to the training process, leading to slower convergence for agents and reduced training stability \cite{papoudakis2021benchmarking}. Therefore, a more attractive solution lies in optimizing the exploration strategy to better guide agents toward optimal joint action, without sacrificing network simplicity.

\textbf{Advanced Exploration Strategies in MARL:} Several approaches have focused on enhancing exploration strategies to help agents gain a more comprehensive understanding of the environment, thereby enabling more accurate value estimation. Curiosity-based approaches \cite{bohmer2019exploration,liu2021cooperative} encourage exploration through intrinsic motivation methods; EITI \cite{wang2019influence} and VM3-AC \cite{kim2020maximum} encourage exploration by maximizing the information obtained by agents. Although these approaches can effectively explore state or action patterns within the environment, they neglect the guiding role of exploration on the policy. Consequently, when agents get trapped in suboptimal solutions, these exploration strategies lack the capability to correct the training direction. LH-IRQN \cite{lyu2018likelihood} and DFAC \cite{sun2021dfac} attempt to estimate uncertainty in multi-agent environments using classifications or quantile distributions. However, computing the distribution is challenging due to the mutual influence among agents. EMAX \cite{schaefer2023ensemble} equips each agent with multiple sub-networks to simplify mean and variance estimation, thereby facilitating the exploration of maximum potential rewards. However, this architectural choice inevitably increases the computational overhead of the algorithm. Therefore, there is a pressing need in multi-agent reinforcement learning for a lightweight exploration strategy equipped with global optimal guidance to boost their performance.

\textbf{Guidance-driven Policy Learning:} With the aim of overcoming suboptimal solutions, some prior proposed methods leverage optimism to guide the direction of policy training. ParetoAC \cite{christianos2023pareto} optimistically utilizes the highest expected return among all possible teammate action combinations as the update target. However, evaluating these combinations is computationally intractable in environments with many agents, limiting its scalability. WQMIX \cite{rashid2020weighted} integrates optimism into value decomposition by assigning lower weights to suboptimal joint actions in the loss function. While this helps identify optimal joint actions, our experiments demonstrate that such direct constraints on value updates restrict the policy during early training, degrading convergence speed. OVI-QMIX \cite{li2024optimistic} avoids directly constraining the policy by introducing optimistic instructors, which guide updates by minimizing the KL divergence between the current and optimal policies. Several single-agent studies have focused on optimizing exploration strategies. Decoupled RL \cite{schafer2022decoupled} leverages intrinsic rewards while ensuring a stable training process by decoupling the exploration and execution policies. Preference-Guided Stochastic Exploration \cite{huang2023sampling} introduces an additional exploration strategy based on the specific advantage of each action. Distributional exploration \cite{mavrin2019distributional} estimates the upper quantile of the expected return to design exploration rewards, encouraging agents to explore actions likely to yield high returns. Although effective in single-agent environments, these methods struggle to extend to multi-agent settings, as the mutual influence among agents vastly complicates the estimation of policy and return distributions. Consequently, designing a CTDE-compatible exploration strategy capable of effectively guiding the training process toward optimal joint actions remains an open challenge.

\section{Preliminaries}\label{sec:3}
\subsection{Dec-POMDP}
We consider a fully cooperative multi-agent learning problem modeled as a Decentralized Partially Observable Markov Decision Processes (Dec-POMDPs) \cite{albrecht2024multi}. A Dec-POMDP is defined by the tuple $\langle I,S,A,P,R,Z,O,\gamma \rangle$, where $I$ represents the set of agents with cardinality $|I|=n$, and $S$ denotes the set of global states. $A$ represents the set of all possible actions for each agent. At each time step, each agent $i \in I$ chooses an action $a_i \in A$. The state transition function $P(s'|s,a_1,...,a_n):S \times A^n \times S \rightarrow [0,1]$ describes the transition probability distribution over global states $S$. The reward function $R(s,a_1,...,a_n):S\times A^n \rightarrow \mathbb{R}$ specifies the reward from the environment. $Z$ represents the set of possible local observations for the agents, where each individual observation $o_i \in Z$ is drawn according to the observation distribution function $O(s,i): S \times I \times Z \rightarrow [0,1]$. Each agent has an observation-action history $\tau_i \in (Z\times A)^*$. Agents aim to find a joint policy $\pi(a_1,a_2,...,a_n|\tau_1,\tau_2,...,\tau_n) = \{\pi_1(a_1|\tau_1),...,\pi_n(a_n|\tau_n)\}$ that maximizes the discounted cumulative return $R_{tot}=\sum_{t=0}^{\infty} \gamma^t R_t(s,a_1,...,a_n)$, where $\gamma \in [0,1]$ is the discount factor over time.

\subsection{Value Decomposition}

The value decomposition method implements the CTDE paradigm in value-based multi-agent reinforcement learning. This approach introduces
a global action-value function $Q_{tot}(\mathbf{\tau},\mathbf{a})$ for the entire team, along with local action-value functions $Q_{i}(\tau_i,a_i)$ for each agent, where $\tau$ and $\tau_i$ represent joint and individual action-observation histories respectively. VD methods assume that the global function can be factored into the local functions through an aggregation function, expressed as $Q_{tot}(\mathbf{\tau},\mathbf{a})=g(Q_{1}(\tau_1,a_1),...,Q_{n}(\tau_n,a_n))$. To ensure consistency between global and local estimations, the aggregation function  $g(\cdot)$ should satisfy the Individual-Global-Max (IGM) condition. The mathematical representation of the IGM condition is as follows:

\begin{equation} \label{eq:1}
    \mathop {arg\max} \limits_{a}Q_{tot}\left( {\tau },{a} \right) =\left( \begin{array}{c}
        \mathop {arg\max}_{a_1}Q_1\left( \tau _1,a_1 \right)\\
        \vdots\\
        \mathop {arg\max}_{a_n}Q_n\left( \tau _n,a_n \right)\\
    \end{array} \right) 
\end{equation}

An aggregation function that satisfies the IGM condition ensures that when the team's joint action maximizes the global action-value function, the local action-value functions of each agent are maximized as well. This allows each agent to achieve the joint optimal policy by greedily selecting the action that maximizes its local value, enabling decentralized execution.

\subsection{$\epsilon$-Greedy Exploration}

$\epsilon$-Greedy Exploration is a straightforward stochastic strategy widely used to balance the exploration-exploitation trade-off in reinforcement learning \cite{sutton1998reinforcement}. It introduces a parameter $\epsilon \in [0,1]$ to govern this balance: with probability $\epsilon$, an agent randomly selects an action from the available action space to encourage exploration; conversely, with probability $1-\epsilon$, the agent exploits its current knowledge by choosing the action that yields the highest expected return. It is worth noting that in practice, to ensure a deterministic greedy choice when multiple actions share the same maximum value, algorithms typically designate a unique optimal action by selecting the one with the smallest index. To align with this practical paradigm, the $\epsilon$-Greedy action selection strategy is mathematically formulated based on the existing value estimates as follows:

\begin{equation} \label{eq:3}
    \begin{aligned}
        \pi \left( a|s \right) =\begin{cases}
            1-\epsilon \,\,+\frac{\epsilon}{|a|}&		if\,\,a\,\,=\,\,arg\max_a Q\left( s,a \right)\\
            \frac{\epsilon}{|a|}&		if\,\,a\,\,\ne \,\,arg\max_a Q\left( s,a \right)\\
        \end{cases}
    \end{aligned}
\end{equation}

\section{Methodology}\label{sec:4}

In this section, we analyze parameter optimization in conventional monotonic value decomposition methods, demonstrating that the temporal-difference (TD) loss associated with each joint action is implicitly weighted by the joint action's sampling frequency during exploration; thus, insufficient sampling of optimal joint actions causes systematic value underestimation. To sample optimal joint actions more frequently during exploration, we introduce monotonically non-decreasing \textbf{optimistic action-value networks} trained to estimate the maximum expected returns. The update mechanism of these networks is termed ``optimistic,'' as they incorporate new rewards only when they exceed current estimates to safely track the upper bound of expected returns. We prove theoretically that these optimistic action-value networks converge in probability to the maximum expected returns for each agent's actions, identifying optimal joint actions through simple return comparisons. Based on this analysis, we propose the \textbf{Optimistic $\epsilon$-Greedy Exploration} strategy. By using the optimistic action-value networks' estimations to select actions during exploration, this strategy effectively guides the policy toward optimal joint actions.

\subsection{Analysis of Parameter Optimization}

In this subsection, we analyze the fundamental cause of the underestimation of optimal joint actions in conventional monotonic value decomposition methods to motivate our proposed approach. While this problem is typically viewed as an architectural bottleneck, we reveal that it is also heavily driven by suboptimal exploration outcomes. Through an analysis of the parameter optimization process for the joint action-value network, we demonstrate how exploration strategies implicitly weight the TD loss during training, leading to an under-emphasis on fitting the values of these optimal joint actions. Furthermore, we illustrate the impact of this under-emphasis on individual agents, showing how the exploratory behavior of teammates inadvertently forces a given agent to underestimate the value of its corresponding optimal action.

In practical value-based multi-agent reinforcement learning, neural networks optimize their parameters $\theta$ by sampling batches of trajectories from an experience replay buffer to minimize the TD loss. This parameter optimization process during each training step is formally expressed as:

\begin{equation} \label{eq:5}
\theta =\underset{\theta}{arg\min}\frac{1}{|B|}\cdot\sum_{\textcolor{black}{\tau \in B}}{[Q_{tot}(\mathbf{\tau},\mathbf{a};\theta) - y_{target}]^2} 
\end{equation}

Since the trajectories stored in the replay buffer are collected by the exploration strategy, and their quantity is relatively small compared to the total number of trajectories generated throughout the entire training process, we can approximate that the distribution of the sampled joint actions aligns with the current exploration strategy $\pi(\mathbf{a}|\mathbf{\tau})$. Consequently, the expected frequency of each joint action within a sampled batch is approximately $|B|\cdot\pi(\mathbf{a}|\mathbf{\tau})$. To demonstrate the relationship between the loss function and the joint actions, we group the loss terms by distinct joint actions rather than individual trajectories. This allows us to substitute the summation $\sum_{\tau \in B}$ in Eq. \eqref{eq:5} with $\sum_{\mathbf{a}\in A^n} |B|\pi(\mathbf{a}|\mathbf{\tau})$, reformulating the TD loss to reveal an implicit weighting mechanism as follows:

\begin{equation} \label{eq:6}
    \theta =\underset{\theta}{arg\min}\sum_{\mathbf{a}\in A^n}{\pi(\mathbf{a}|\mathbf{\tau})[Q_{tot}(\mathbf{\tau},\mathbf{a}) -  y_{target}]^2} 
\end{equation}

Eq. \eqref{eq:6} exposes a potential flaw in the parameter optimization process. It reveals that the loss function for any joint action is implicitly weighted by its probability of being sampled during exploration. Consequently, if the agents insufficiently sample optimal joint actions during exploration, the neural networks may fail to accurately fit their true values, shifting their focus instead toward optimizing the values of suboptimal joint actions. This flaw introduces a significant risk of value underestimation. To understand how this implicit weighting impacts individual agents within conventional value decomposition methods, we further decompose this joint probability. Since each agent selects its actions independently during exploration, the joint policy can be factorized as $\pi(\mathbf{a}|\mathbf{\tau}) = \prod_{i=1}^n{\pi(a_i|\tau_i)}$. For any given agent $i$, we can isolate its perspective by rewriting the optimization objective:

\begin{equation} \label{eq:7}
    \begin{aligned}
        \theta =&\underset{\theta}{arg\min}\sum_{\mathbf{a}\in A^n}{\pi (\mathbf{a}|\mathbf{\tau })L(\mathbf{\tau },a_i,a_{-i})}\\
         =&\underset{\theta}{arg\min}\sum_{\mathbf{a}\in A^n}{\pi (a_i|\tau _i)\pi (a_{-i}|\tau _{-i})}L(\mathbf{\tau },a_i,a_{-i})\\
         =&\underset{\theta}{arg\min}\sum_{a_i\in A}{\pi (a_i|\tau _i)\sum_{a_{-i}\in A^{n-1}}{\pi (a_{-i}|\tau _{-i})L(\mathbf{\tau },a_{i},a_{-i})}}\\
    \end{aligned}
\end{equation}

Here $L(\mathbf{\tau},a_{i},a_{-i}) = [Q_{tot}(\mathbf{\tau },a_{i},a_{-i})-y_{target}]^2$ for short. Eq. \eqref{eq:7} provides a critical insight into the parameter optimization process: for a given action $a_i$ of agent $i$, the optimization of its value estimation is heavily influenced by the exploration strategies of all other agents ($\pi (a_{-i}|\tau _{-i})$). When other agents frequently select suboptimal actions during exploration, the corresponding loss terms gain significantly larger implicit weights. Consequently, the joint action-value network is forced to pay more attention to minimizing the estimation error for these low-return scenarios. In conventional monotonic value decomposition methods, the given agent's value estimation for this action is forced to fit the returns of scenarios where other agents choose uncoordinated actions, leading to an underestimation of the action's true value.

To resolve this deficiency, the parameter optimization process must be realigned to focus on optimal joint actions. Based on the mechanics revealed in Eq. \eqref{eq:7}, a highly effective solution is to increase the probability of selecting optimal joint actions during the exploration phase as much as possible. By maximizing the sampling probabilities of these optimal joint actions, we correspondingly increase their implicit weights within the loss function, thereby forcing the parameter optimization process to pay significantly more attention to optimal joint actions when minimizing the TD loss.

\subsection{Optimistic Update}\label{sec:4.2}

In this subsection, we address the crucial step for sampling optimal joint actions more frequently during exploration: successfully identifying them. To achieve this, we introduce \textbf{optimistic action-value networks} that estimate the maximum expected returns and are updated in a monotonically non-decreasing manner. The nature of these networks is ``optimistic,'' as they consistently incorporate higher returns into their estimates while ignoring lower values. We prove theoretically that estimation functions representing these networks converge in probability to the maximum achievable returns under the Dec-POMDP setting, thereby serving as reliable indicators of optimal joint actions. By establishing this theoretical foundation, we demonstrate how agents can accurately pinpoint optimal joint actions despite the noisy and penalized returns caused by uncoordinated teammates.

To facilitate our subsequent theoretical analysis, we first formalize the core mathematical property that underpins the optimistic action-value networks. We define an estimation function's update sequence as being optimistic with respect to a target constant $c$ if and only if it approaches this constant in a monotonically non-decreasing manner. Formally, this property is defined as follows:

\begin{definition}\label{def:1}
Given an input $x$, a constant $c$, and an ordered sequence of functions $\{f_1(x), f_2(x), \dots, f_n(x)\}$ such that $f_1(x) \leq c$, this sequence is defined as updating optimistically toward $c$ if and only if $\forall i$, $f_i(x) \leq c$, and $\forall i < j$, $f_i(x) \leq f_j(x)$.
\end{definition}

Intuitively, optimal joint actions yield the highest possible returns, motivating us to approximate this upper bound by updating our estimations optimistically. However, during the exploration process, we need to determine the optimal action for each individual agent. To shift our focus from joint actions to those of individual agents, we employ a VDN-like additive decomposition for the global return $R(s,a)=\sum{r_i(s,a_i|a_{-i})}$. Benefiting from the IGM property of this additive decomposition, it holds that $\max_{a_{-i}} r_i(s,a_i|a_{-i}) < \max_{a_{-i}} r_i(s,a_i^*|a_{-i})$ for any suboptimal action $a_i\neq a_i^*$. In other words, the maximum value of the return decomposition terms for suboptimal actions is always smaller than the corresponding terms for the true optimal action $a_i^*$. This property serves as a reliable indicator for identifying optimal joint actions' components. We introduce a set of optimistic estimation functions $\{f_i(x)\}$ which our networks approximate to track the maximum value of the return decomposition terms for each agent. For a given input $x=a_i$, the function is updated as follows:

\begin{equation} \label{eq:10}
    \begin{aligned}
        f_i^{t+1}\left( x \right)= \begin{cases}
            f_i^t\left( x \right) +\alpha \left( r_i^t-f_i^t\left( x \right) \right)&		if\,\,r_t>f_i^t\left( x \right)\\
            f_i^t\left( x \right)&		else\\
        \end{cases}
    \end{aligned}
\end{equation}

$\alpha\in [0,1]$ is a predefined hyperparameter representing the learning rate for the estimation function $f_i(x)$, and $r_i^t$ denotes the return decomposition term $r_i(s,a_i|a_{-i})$ at timestep $t$. Because the mathematical properties of these functions are identical across all agents, we simplify $f_i^t$ and $r_i^t$ to $f_t$ and $r_t$ for notational simplicity in the subsequent analysis. Note that the update process defined in Eq. \eqref{eq:10} satisfies the condition for updating optimistically as outlined in Definition \ref{def:1}. This is because each step either strictly increases the estimate when a higher return is observed or leaves it unchanged, ensuring a monotonically non-decreasing sequence. The upper and lower bounds of the estimation function presented in Eq. \eqref{eq:10} can be established using the following lemma:

\begin{lemma}\label{lemma:1}
For the estimation function sequence $\{f_t(x)\}$ defined in Eq. \eqref{eq:10}, if initialized with $f_0(x)\leq r_{max}$, where $r_{max}=\max_t r_t$, then $\forall t\geq0$, it holds that $f_t(x)\leq r_{max}$.
\end{lemma}

\begin{lemma}\label{lemma:2}
    Given an auxiliary function sequence updating in the following form:
    \begin{equation*}
	\begin{aligned}
		f_{t+1}^{'}\left( x \right)= \begin{cases}
		      f_t^{'}\left( x \right) +\alpha \left( r_t-f_t^{'}\left( x \right) \right)&		if\,\,r_t=r_{max}\\
				f_t^{'}\left( x \right)&		else\\
		\end{cases}
	\end{aligned}
    \end{equation*}
    For the estimation function sequence $\{f_t(x)\}$ defined in Eq. \eqref{eq:10}, if initialized with $f_0(x) = f_0^{'}(x)$, then $\forall t\geq 0$, it holds that $f_t(x)\geq f_t^{'}(x)$.
\end{lemma}

Lemma \ref{lemma:1} establishes an upper bound for the estimation function $f_t(x)$, as each update term never exceeds the difference between the current estimate and the maximum return decomposition term. Conversely, Lemma \ref{lemma:2} establishes a lower bound: because the function incorporates any observed return greater than the current estimate, it increases at least as fast as an auxiliary sequence that only updates when the absolute maximum return is observed. Formal proofs for both lemmas are provided in the Appendix \ref{secA}.

The lower bound of $f_t(x)$ is represented using an auxiliary function sequence $f_t^{'}(x)$. To analyze its convergence properties, we evaluate the mathematical expectation of this sequence. Let event $A$ denote the scenario where, during the exploration process, all other agents precisely choose their corresponding optimal action components, yielding the maximum return $r_t = r_{max}$. The probability of this event occurring, given that agent $i$ chooses action $x=a_i$, is expressed as $P(a_{-i}=a_{-i}^{*}|x=a_i)=c$, where $c > 0$. Let $\mathbb{I}(A)$ denote the cumulative number of times event $A$ occurs, and let $f_{\mathbb{I}(A)=k}^{'}(x)$ represent the value of the sequence $f^{'}(x)$ after event $A$ has occurred exactly $k$ times. For $n$ occurrences of event $A$, the closed-form expression for $f_{\mathbb{I}(A)=n}^{'}(x)$ can be derived as follows:

 \begin{equation}\label{eq:11}
    \begin{aligned}
        f_{\mathbb{I}(A)=n}^{'}(x)&=f_{\mathbb{I}(A)=n-1}^{'}(x) + \alpha(r_{max}-f_{\mathbb{I}(A)=n-1}^{'}(x))\\
        &=(1-\alpha)f_{\mathbb{I}(A)=n-1}^{'}(x)+\alpha r_{max}\\
        &=(1-\alpha)^2f_{\mathbb{I}(A)=n-2}^{'}(x)+[\alpha+\alpha(1-\alpha)]r_{max}\\
        &=(1-\alpha)^nf_{\mathbb{I}(A)=0}^{'}(x)+\sum_{k=0}^{n-1}{\alpha(1-\alpha)^kr_{max}}\\
        &=(1-\alpha)^nf_{\mathbb{I}(A)=0}^{'}(x)+\alpha\frac{1-(1-\alpha)^n}{1-(1-\alpha)}r_{max}\\
        &=r_{max}+(1-\alpha)^n[f_0^{'}(x)-r_{max}]
    \end{aligned}
\end{equation}
 
With the help of the closed-form expression of $f_{\mathbb{I}(A)=n}^{'}(x)$. The expectation of $f_t^{'}(x)$ can be computed as follows:

\begin{equation}
    \begin{aligned}
        E[f_t^{'}(x)]&=\sum_{n=0}^t{P(\mathbb{I}(A)=n)f_{\mathbb{I}(A)=n}^{'}(x)}\\
        &=\sum_{n=0}^t{\binom{t}{n}  c^n\left( 1-c \right) ^{t-n}\left( r_{max}+(1-\alpha)^n[f_0^{'}(x)-r_{max}]\right)}\\
        &=r_{max}\sum_{n=0}^t{\binom{t}{n}  c^n\left( 1-c \right) ^{t-n}}+[f_0^{'}(x)-r_{max}]\sum_{n=0}^t{\binom{t}{n} (c-c\alpha)^n\left( 1-c \right) ^{t-n}}\\
        &=r_{max} + [f_0^{'}(x)-r_{max}](1-c\alpha)^t
    \end{aligned}
\end{equation}

Utilizing Lemma \ref{lemma:1} and Lemma \ref{lemma:2}, along with the mathematical expectation of Eq. \eqref{eq:11} and Markov's Inequality, we establish the following theorem:

\begin{theorem}\label{theorem:1}
    Consider the sequence of functions $\{f_t(x)\}$, initialized with $f_0(x)\leq r_{max}$ and updated according to Eq. \eqref{eq:10}. It holds that $\{f_t(x)\}$ converges in probability to $r_{max}$, i.e.,$f_t\left( x \right) \xrightarrow{p}r_{max}$ as $t\rightarrow\infty$.
\end{theorem}

\begin{proof}
    $\forall \varepsilon > 0$, we have
    \begin{equation}
        \begin{aligned}
            P(|r_{max}-f_t(x)|\geq \varepsilon) &= P(r_{max}-f_t(x)\geq \varepsilon)\\
            &\leq P(r_{max}-f_t^{'}(x)\geq \varepsilon)\\
            &\leq \frac{E[r_{max}-f_t^{'}(x)]}{\varepsilon}\\
            &=\frac{r_{max}-E[f_t^{'}(x)]}{\varepsilon}\\
            &=\frac{[f_0^{'}(x)-r_{max}](1-c\alpha)^t}{\varepsilon}\\
            &\xrightarrow{t\rightarrow +\infty}0
        \end{aligned}
    \end{equation}

    Based on the definition of convergence in probability, we have that $f_t(x)$ converges in probability to $r_{max}$, i.e., $f_t\left( x \right) \xrightarrow{p}r_{\max}$.
    
\end{proof}

Theorem \ref{theorem:1} demonstrates that the sequence of estimation functions updated according to Eq. \eqref{eq:10} converges in probability to the maximum value $\max_{a_{-i}} r_i(a_i|a_{-i})$ for any given action $a_i\in A$. This convergence enables the precise identification of the optimal action component $a_i^*$ by comparing the function values across different actions. However, since the team only receives the total reward as feedback in practice, we exploit the additivity property of convergence in probability. Specifically, we update the sum of all individual optimistic action-value networks toward the total cumulative team reward. This architectural design ensures that our Optimistic $\epsilon$-Greedy Exploration strategy can guide the joint policy toward optimal joint actions by maximizing the frequency of selecting actions that yield the highest expected returns.

\subsection{Overall Framework}

\begin{figure*}[t!]
    \centering
    \includegraphics[width=0.95\textwidth]{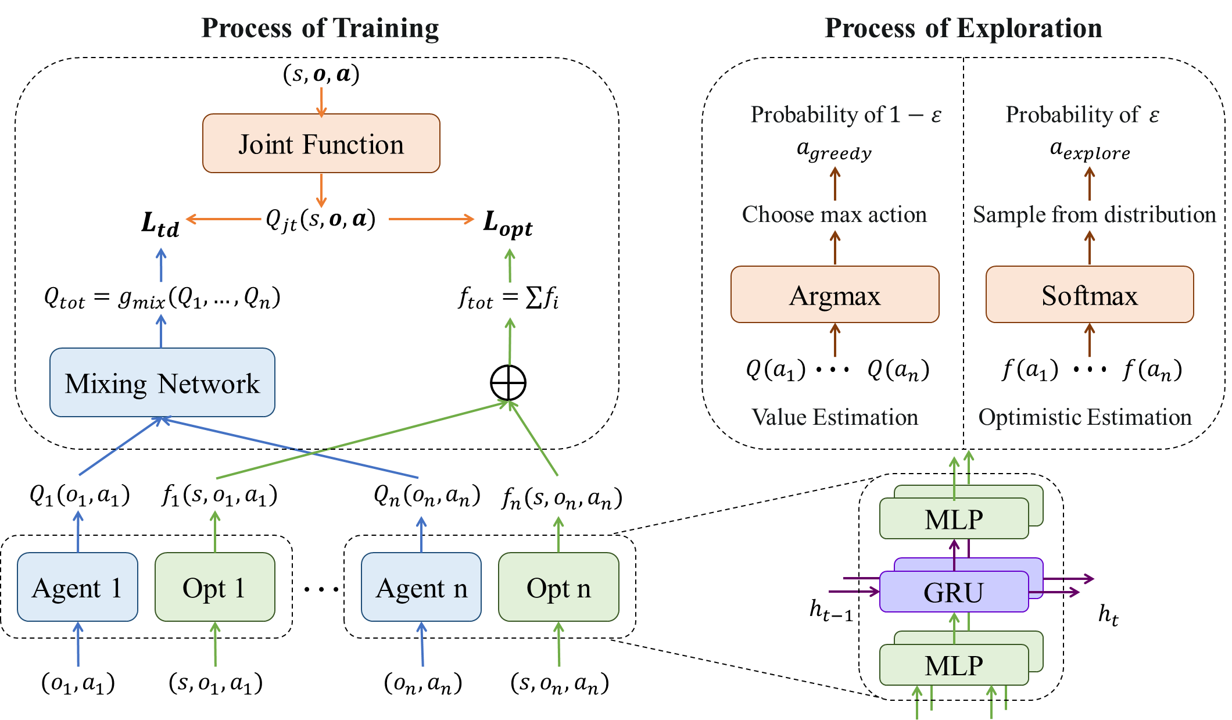}
    \caption{Overall framework of Optimistic $\epsilon$-Greedy Exploration}
    \label{fig:1}
\end{figure*}

In this subsection, we propose the \textbf{Optimistic $\epsilon$-Greedy Exploration strategy} based on the theoretical properties of the optimistic action-value networks established in Section \ref{sec:4.2}. Similar to conventional $\epsilon$-greedy exploration, this strategy uses a probability of $1-\epsilon$ to select the greedy actions, which are the actions estimated to have the highest value by the main action-value networks. The key difference lies in the exploration phase triggered by the probability $\epsilon$: instead of selecting actions uniformly at random, our strategy samples actions based on the outputs of the optimistic action-value networks. This mechanism aims to prioritize the selection of each agent's optimal action, effectively guiding the joint policy toward optimal joint actions. The overall framework of our method is illustrated in Figure \ref{fig:1}.

The right side of Figure \ref{fig:1} depicts the data flow and network architecture of the Optimistic $\epsilon$-Greedy Exploration strategy. Each agent is equipped with two distinct neural networks: the \textbf{Agent} network and the \textbf{Opt} network. The \textbf{Agent} network computes the standard expected return $Q_i(o_i,a_i)$ for each agent, while the \textbf{Opt} network calculates the optimistic action-value estimation $f_i(s,o_i,a_i)$. Both networks take the agent's local observation as input. Crucially, we additionally incorporate the global state $s$ into the input of the Opt network to assist it in robustly identifying the maximum achievable returns. Internally, an MLP first extracts features from the inputs, followed by a GRU to integrate historical trajectory information into the hidden representation. Finally, this representation is processed by another MLP to output the final value estimations for each available action.

During exploration, each agent selects the greedy action $\hat{a}^*$ which maximizes $Q_i(o_i,a_i)$ with a probability of $1-\epsilon$, and samples exploratory actions based on the outputs $f_i(s,o_i,a_i)$ of its Opt network with a probability of $\epsilon$.  While any monotonic function can theoretically be used to map these optimistic estimations to a distribution over the action space, we adopt the Softmax function for this purpose. Specifically, recognizing that the scale of expected returns can vary significantly across different environments, we introduce a temperature normalization mechanism, denoted as $tempNorm(\cdot)$, to $f_i(s,o_i,a_i)$ to enhance its robustness. During this normalization process, all values are proportionally scaled to a range $[0, \beta]$, where the scaling bound $\beta$ is initialized to 0 and linearly increases over time to a predefined hyperparameter. This approach enables the algorithm to perform near-uniform random exploration during the initial phases of training, effectively circumventing the early-stage instability inherent to neural networks. As the Opt networks converge and become more adept at identifying the optimal action of each agent, this normalized range is progressively annealed. This ensures that these optimal actions are sampled with increasingly higher frequencies. In summary, each agent $i$ determines its action selection probabilities during exploration according to the following formulation:

\begin{equation} \label{eq:12}
    \begin{aligned}
        \pi_i \left( a_i|s \right) =\begin{cases}
            1-\epsilon \,\,+\epsilon\cdot Softmax(tempNorm(f_i(a_i)))&		if\,\,a_i=\hat{a}_i^*\\
            \epsilon\cdot Softmax(tempNorm(f_i(a_i)))&		if\,\,a\ne \hat{a}_i^*\\
        \end{cases}
    \end{aligned}
\end{equation}

\vspace{\baselineskip}

Note that when Eq. \eqref{eq:12} is employed as the exploration strategy, the following theorem holds:

\begin{theorem}
The Optimistic $\epsilon$-Greedy Exploration strategy defined in Eq. \eqref{eq:12} samples the optimal action of each agent with a higher probability than the traditional $\epsilon$-greedy exploration strategy defined in Eq. \eqref{eq:3}.
\end{theorem}

Intuitively, while both strategies share the same probability of choosing the optimal actions of each agent when exploiting the maximum estimated values, the Optimistic $\epsilon$-Greedy Exploration strategy yields a higher probability of picking the optimal actions when actions are selected through sampling. A formal proof of this theorem is provided in the Appendix \ref{secA}.

The left side of Figure \ref{fig:1} illustrates the training process for these neural networks. During the forward computation, the individual values $Q_i(o_i,a_i)$ are aggregated into the global action-value estimation $Q_{tot}(\mathbf{o},\mathbf{a})$ through a Mixing network, which can be any monotonic aggregation function susceptible to the value underestimation problem, as formulated in Eq. \eqref{eq:add1}. Meanwhile, the individual optimistic action-value estimations $f_i(s,o_i,a_i)$ are aggregated strictly through summation to produce the global optimistic estimation $f_{tot}(s,\mathbf{o},\mathbf{a})$, as defined in Eq. \eqref{eq:add2}:

\begin{equation}\label{eq:add1}
    Q_{tot}(\mathbf{o},\mathbf{a}) = g(Q_1(o_1,a_1),...,Q_n(o_n,a_n))
\end{equation}

\begin{equation}\label{eq:add2}
    f_{tot}(s,\mathbf{o},\mathbf{a}) = \sum_{i=1}^n{f_i(s,o_i,a_i)}
\end{equation}

The parameters of the \textbf{Agent} networks are optimized by minimizing the standard TD loss, as shown in Eq. \eqref{eq:13}. To ensure that the \textbf{Opt} networks are updated optimistically in a monotonically non-decreasing manner, a specifically designed weighted optimistic loss function $L_{opt}$ is employed, as defined in Eq. \eqref{eq:14}. Specifically, the weight $w$ is set to 1 when $y_{target} > f_{tot}(s,\mathbf{o},\mathbf{a})$, and to 0 when $y_{target} \leq f_{tot}(s,\mathbf{o},\mathbf{a})$. This masking mechanism ensures the loss function is masked whenever the optimistic estimation exceeds the target, thereby enabling monotonic non-decreasing updates. However, considering that the $y_{target}$ includes neural network-based estimations, the weight $w$ is set to a hyperparameter instead of 0 in practical implementation to enhance the robustness of the optimistic action-value network.

\begin{equation} \label{eq:13}
L_{td}=[Q_{tot}(\mathbf{o},\mathbf{a})-y_{target}]^2
\end{equation}
	
\begin{equation} \label{eq:14}
L_{opt}=w[f_{tot}(s,\mathbf{o},\mathbf{a})-y_{target}]^2
\end{equation}

\vspace{\baselineskip}

Building on this foundation, we introduce an unconstrained neural network to better estimate the expected return. This network takes the state, observation, and joint action as inputs, estimating the expected return $Q_{jt}(s,\mathbf{o},\mathbf{a})$ under the given conditions. Such a design enables the neural network to explore a broader parameter space to map inputs to the expected return. The TD target is calculated using Eq. \eqref{eq:15}. Inspired by the Double Q-learning approach, we leverage the estimations corresponding to the optimal joint action $\hat{\textbf{a}}_{t+1}^* = \arg\max_{\textbf{a}_{t+1}}Q_{tot}(\mathbf{o}_{t+1},\mathbf{a}_{t+1})$ to approximate the maximum return for the next state. This design avoids an exhaustive search over the exponentially large joint action space. $Q_{jt}(s,\mathbf{o},\mathbf{a})$ is updated using the TD loss defined in Eq. \eqref{eq:16}.

\begin{equation} \label{eq:15}
    y_{target} = \textcolor{black}{r_t}+\gamma Q_{jt}(s_{t+1},\mathbf{o}_{t+1},\hat{\textbf{a}}_{t+1}^*)
\end{equation}
	
\begin{equation} \label{eq:16}
    L_{jt}=[Q_{jt}(s,\mathbf{o},\mathbf{a})-y_{target}]^2
\end{equation}

\vspace{\baselineskip}

The overall objective of parameter optimization is to minimize the combined loss defined in Eq. \eqref{eq:13}, Eq. \eqref{eq:14}, and Eq. \eqref{eq:16}, as formulated in Eq. \eqref{eq:17}.

\begin{equation} \label{eq:17}
    L=L_{td}+L_{opt}+L_{jt}
\end{equation}

\vspace{\baselineskip}

\textbf{Discussion:} Our notion of optimism is closely related to prior works such as ParetoAC and WQMIX. By introducing monotonically non-decreasing optimistic action-value networks to estimate maximum expected returns, our exploration strategy equips conventional monotonic value decomposition methods with the ability to focus on optimal joint actions without requiring complex architectural modifications. Compared to ParetoAC, which requires evaluating all possible joint action combinations of other agents to compute Pareto optimality, resulting in an exponential computational complexity of $\mathcal{O}(|A|^n)$, our approach identifies the optimal action components of optimal joint actions for each individual agent through mathematically guaranteed monotonically non-decreasing networks. This reduces the complexity to $\mathcal{O}(|A|)$ per agent, achieving significantly better scalability. Furthermore, compared to WQMIX, which explicitly applies optimistic weighting directly to the loss function of the main action-value network, our approach applies optimism exclusively to independent networks used solely to determine action selection probabilities during exploration. This decouples exploration from the execution policy, mitigating instability during the training process and thereby accelerating convergence. Finally, compared to conventional exploration strategies, such as $\epsilon$-greedy or noise-based exploration, our Optimistic $\epsilon$-Greedy Exploration strategy guides agents with more accurate targets, equipping the framework with the capability to escape from suboptimal joint actions. In light of these theoretical advantages, the subsequent experimental section will empirically validate how integrating conventional monotonic value decomposition with our Optimistic $\epsilon$-Greedy Exploration strategy mitigates the shortcomings of existing approaches. Specifically, we will demonstrate that our method can successfully identify optimal joint actions, thereby achieving higher final returns, superior win rates, and faster convergence, all while retaining algorithmic simplicity.

\section{Experimental Evaluation}\label{sec:5}

In this section, we integrate the \textbf{Optimistic $\epsilon$-Greedy Exploration} strategy into QMIX and VDN, two representative conventional monotonic value decomposition algorithms, creating enhanced versions referred to as OPT-QMIX and OPT-VDN, respectively. To evaluate the efficacy of these enhanced algorithms, we conduct experiments across multiple environments, aiming to address the following key research questions:

\begin{itemize}
    \item \textbf{Question (1):} Does the Optimistic $\epsilon$-Greedy Exploration strategy effectively resolve the value underestimation problem inherent in conventional monotonic value decomposition frameworks?  
    \item \textbf{Question (2):} How does the Optimistic $\epsilon$-Greedy Exploration strategy compare with methods that enhance network architectures or employ manual weighting schemes, particularly in complex environments? 
    \item \textbf{Question (3):} Can the Optimistic $\epsilon$-Greedy Exploration strategy outperform other widely used exploration strategies in existing MARL frameworks?  
\end{itemize}

We evaluate our approach across three environments of escalating complexity, each presenting inherent value underestimation challenges: Matrix Games in Section \ref{subsection5.2}, Predator and Prey in Section \ref{subsection5.3}, and the StarCraft Multi-Agent Challenge (SMAC) in Section \ref{subsection5.4}. To demonstrate the fundamental effectiveness of our method, we first compare our variants, OPT-QMIX and OPT-VDN, against their foundational baselines, QMIX \cite{rashid2020monotonic} and VDN \cite{sunehag2017value}. Subsequently, we benchmark against advanced state-of-the-art methods, specifically OW-QMIX/CW-QMIX \cite{rashid2020weighted} and QTRAN \cite{son2019qtran}. OW-QMIX and CW-QMIX are selected because they also employ weighted updates but apply them directly to the TD loss of the main action-value network, whereas QTRAN is included to represent approaches that utilize complex network architectures to construct lower bounds. Furthermore, an ablation study in Section \ref{subsection5.5} leverages the QMIX framework to compare our strategy against exploration strategies commonly used in other MARL algorithms, including conventional $\epsilon$-greedy \cite{sutton1998reinforcement}, noise-based \cite{wang2022noise}, and intrinsic reward-driven exploration \cite{zhou2025double}. Finally, Section \ref{subsection5.6} presents a qualitative analysis that empirically demonstrates how our method successfully steers the joint policy toward optimal joint actions by increasing the sampling frequencies of these optimal joint actions.

\subsection{Experiment Setup}\label{subsection5.1}

To ensure consistency in code implementation, we utilized the open-source PyMARL framework \cite{samvelyan2019starcraft} to conduct all our experiments. Each experiment was independently run over 5 different random seeds. The median of the results across these seeds is reported as the primary evaluation metric, with the interquartile range (25\%-75\%) visualized as shaded regions in the resulting plots. All experiments were carried out on a server equipped with an Intel Core i5-12400 processor, 64 GB of RAM, and an NVIDIA GeForce RTX 4060 Ti GPU with 16 GB VRAM.

To ensure fair comparisons, all algorithms were configured with identical common hyperparameters, closely following the default settings in PyMARL. Specifically, the capacity of the replay buffer for all algorithms was set to 5,000 episodes during training. At each training step, batches of 32 episodes were sampled from the buffer to train the neural networks. The network parameters were optimized using the RMSprop optimizer with a learning rate of $5\times 10^{-4}$. The discount factor $\gamma$ for calculating the TD target was set to 0.99.

The individual agent network architectures were kept uniform across all algorithms. Each agent's network consisted of an MLP with a hidden dimension of 64, a GRU with a hidden dimension of 64, and a final MLP whose output dimension corresponded to the number of available actions. During evaluation, all algorithms greedily selected actions based on the highest estimated values from their respective agent networks. Unless stated otherwise, all baseline algorithms employed the conventional $\epsilon$-greedy strategy for exploration. To rigorously evaluate the impact of exploration on algorithmic performance, we extended the $\epsilon$-annealing period, following practices from prior work \cite{rashid2020weighted}. Specifically, the value of $\epsilon$ is kept constant at 1.0 in the Matrix Games. In the Predator and Prey environment, $\epsilon$ is linearly annealed from 1.0 to 0.05 over 200k timesteps. In the SMAC environments, $\epsilon$ is linearly annealed from 1.0 to 0.05 over 1M timesteps.

OPT-QMIX, OW-QMIX, and CW-QMIX employ the same unconstrained joint action-value network $Q_{jt}$ to estimate expected returns. For the Matrix Games and the Predator and Prey environment, the weight $w$ in these algorithms is set to 0.01, while $w$ for SMAC is set to 0.5. For our Optimistic $\epsilon$-Greedy Exploration strategy, the temperature normalization range was set to $[0,\beta]$. The scaling bound $\beta$ was linearly increased from 0 to 2 over the first 20k timesteps in the Predator and Prey environment, and from 0 to 1 over the first 50k timesteps in SMAC. Because the Matrix Games serve as simplified toy environments, temperature normalization was not employed there.

\subsection{Matrix Games}\label{subsection5.2}

\begin{table*}[t]
    \centering
    \setlength{\tabcolsep}{5.5pt}
    \begin{subtable}[c]{0.25\textwidth}
        \centering
            \begin{tabular}{|c|c|c|c|}
				\hline
				    & $a_1$ & $a_2$ & $a_3$ \\
				\hline
				$a_1$ & 8 & -12 & -12 \\
				\hline
				$a_2$ & -12 & 0 & 0 \\
				\hline
				$a_3$ & -12 & 0 & 0 \\
				\hline
            \end{tabular}
            \caption{Payoff Matrix}
    \end{subtable}
    \hfill %
    \begin{minipage}[c]{0.7\textwidth}
        \begin{subtable}[b]{0.45\linewidth}
            \centering
            \begin{tabular}{|c|c|c|c|}
				\hline
				  & $a_1$ & $a_2$ & $a_3$ \\
				\hline
				$a_1$ & -3.48 & -3.21 & -3.09 \\
				\hline
				$a_2$ & -3.38 & -3.11 & -2.99 \\
				\hline
				$a_3$ & -2.95 & -2.68 & \textcolor{red}{-2.56} \\
				\hline
            \end{tabular}
            \caption{VDN}
        \end{subtable}
        \hfill
        \begin{subtable}[b]{0.45\linewidth}
            \centering
            \begin{tabular}{|c|c|c|c|}
				\hline
				  & $a_1$ & $a_2$ & $a_3$ \\
				\hline
				$a_1$ & \textcolor{blue}{7.8} & -0.36 & -0.24 \\
				\hline
				$a_2$ & -0.07 & -8.23 & -8.09 \\
				\hline
				$a_3$ & 0.07 & -8.09 & -7.97 \\
				\hline
            \end{tabular}
            \caption{OPT-VDN}
        \end{subtable}

        \vspace{0.2cm} %

        \begin{subtable}[t]{0.45\linewidth}
            \centering
            \begin{tabular}{|c|c|c|c|}
				\hline
				  & $a_1$ & $a_2$ & $a_3$ \\
				\hline
				$a_1$ & -7.76 & -7.76 & -7.76 \\
				\hline
				$a_2$ & -7.76 & \textcolor{red}{0.04} & 0.02 \\
				\hline
				$a_3$ & -7.76 & 0.02 & 0.01 \\
				\hline
            \end{tabular}
            \caption{QMIX}
        \end{subtable}
        \hfill
        \begin{subtable}[t]{0.45\linewidth}
            \centering
            \begin{tabular}{|c|c|c|c|}
				\hline
				  & $a_1$ & $a_2$ & $a_3$ \\
				\hline
				$a_1$ & \textcolor{blue}{7.8} & -6.38 & -6.32 \\
				\hline
				$a_2$ & -6.22 & -6.40 & -6.40 \\
				\hline
				$a_3$ & -6.11 & -6.40 & -6.39 \\
				\hline
            \end{tabular}
            \caption{OPT-QMIX}
        \end{subtable}
    \end{minipage}
    \caption{Value Estimations of Different Joint Actions in Matrix Game\label{table:1}}
\end{table*}

We initiate our experiments with the Matrix Games environment, a simple toy environment consisting of multiple agents and a payoff matrix. Each trajectory spans a single timestep, during which each agent selects an action, and the team receives the reward corresponding to the joint action from the payoff matrix. The Matrix Games environment offers significant flexibility, allowing researchers to simulate various reward distributions by adjusting the payoff matrix. Table \ref{table:1}(a) revisits the matrix game instance used in Section \ref{sec1} to illustrate the problem of systematic value underestimation. Under this distribution, if any agent fails to choose its optimal action (i.e., $a_1$), the team reward decreases sharply. Such a drastic reduction in reward can cause algorithms to underestimate the expected returns of optimal joint actions, leading to convergence on suboptimal policies.

Table \ref{table:1}(b) and Table \ref{table:1}(d) respectively illustrate the value estimations of each action combination for the conventional monotonic value decomposition methods VDN and QMIX. Both methods fail to correctly identify optimal joint actions and instead converge to suboptimal joint actions. In contrast, Table \ref{table:1}(c) and Table \ref{table:1}(e) present the results when our Optimistic $\epsilon$-Greedy Exploration strategy is integrated into VDN and QMIX. The corresponding variants, OPT-VDN and OPT-QMIX, successfully focus their estimations on optimal joint actions. These results demonstrate that our method effectively empowers algorithms to escape suboptimal joint actions and successfully converge toward optimal joint actions.

\begin{table*}[t]
    \centering
    \captionsetup[subtable]{labelformat=empty}
    \begin{tabular}{ccc}
        \begin{subtable}[t]{0.3\textwidth}
            \centering
            \begin{tabular}{|c|c|c|c|}
				\hline
				    & $a_1$ & $a_2$ & $a_3$ \\
				\hline
				$a_1$ & 8 & 0 & 0 \\
				\hline
				$a_2$ & 0 & 0 & 0 \\
				\hline
				$a_3$ & 0 & 0 & 0 \\
				\hline
            \end{tabular}
            \caption{Matrix(a)}
        \end{subtable}
        &
        \begin{subtable}[t]{0.3\textwidth}
            \centering
            \begin{tabular}{|c|c|c|c|}
				\hline
				    & $a_1$ & $a_2$ & $a_3$ \\
				\hline
				$a_1$ & 8 & -12 & -12 \\
				\hline
				$a_2$ & -12 & 0 & 0 \\
				\hline
				$a_3$ & -12 & 0 & 0 \\
				\hline
            \end{tabular}
            \caption{Matrix(b)}
        \end{subtable}
        &
        \begin{subtable}[t]{0.3\textwidth}
            \centering
            \begin{tabular}{|c|c|c|c|}
				\hline
				    & $a_1$ & $a_2$ & $a_3$ \\
				\hline
				$a_1$ & 4 & -12 & -12 \\
				\hline
				$a_2$ & -12 & 0 & 0 \\
				\hline
				$a_3$ & -12 & 0 & 0 \\
				\hline
            \end{tabular}
            \caption{Matrix(c)}
        \end{subtable}
    \end{tabular}
    \caption{Payoff Matrices Used in Experimental Evaluation\label{table:2}}
\end{table*}

\begin{figure*}[t!]
    \centering
    \includegraphics[width=0.95\textwidth]{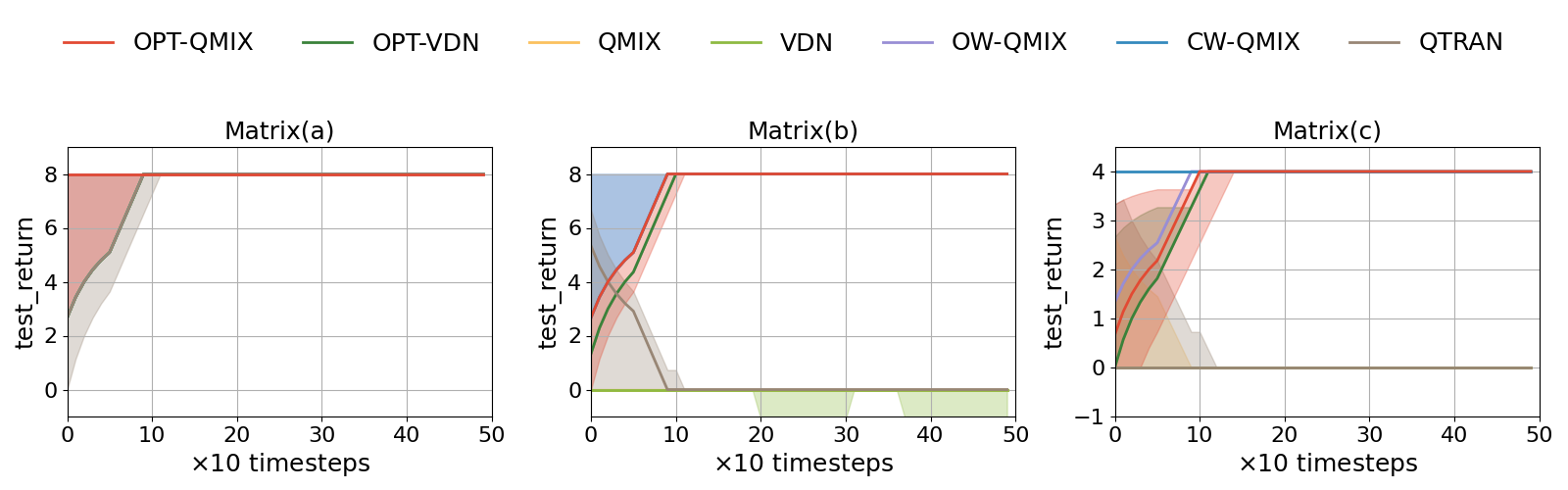}
    \caption{Experimental Results of Matrix Game}
    \label{fig:Matrix}
\end{figure*}

To further investigate the performance of our proposed variants across diverse reward distributions and compare them against other enhanced algorithms, we introduce two additional payoff matrices. The three payoff matrices utilized in our experiments, detailed in Table \ref{table:2}, exhibit progressive levels of convergence difficulty. Specifically, Matrix (a) presents minimal risk of converging to suboptimal joint actions; Matrix (b) introduces negative rewards for uncoordinated actions, creating a risk of convergence to suboptimal joint actions; and Matrix (c) further exacerbates this risk by decreasing the rewards of optimal joint actions. 

Figure \ref{fig:Matrix} illustrates the average test returns of all algorithms across the corresponding matrices over training time. The results demonstrate that our variants successfully converge to optimal joint actions across all three matrices. While the baseline QMIX and VDN algorithms manage to converge on the simpler Matrix (a), they easily fall into suboptimal joint actions on the more challenging Matrices (b) and (c). Regarding the enhanced baseline algorithms, both OW-QMIX and CW-QMIX successfully identify optimal joint actions by explicitly weighting the TD loss. Conversely, although QTRAN is theoretically capable of convergence, its reliance on a complex neural network architecture for lower bound estimation prevents it from converging within the time window depicted in Figure \ref{fig:Matrix}. These empirical results confirm that our method is as effective as weighting-based approaches in guiding the joint policy toward optimal joint actions, while achieving a significantly faster convergence rate than methods reliant on complex architectural modifications.

\subsection{Predator and Prey}\label{subsection5.3}

To evaluate the effects of the \textbf{Optimistic $\epsilon$-Greedy Exploration} strategy on value estimation, we employ the Predator and Prey environment, a classic MARL testbed demanding precise spatial and temporal coordination, as shown in Figure \ref{fig:2}(a). In this setting, predators act as learning agents equipped with $5 \times 5$ local observations and must choose between ``move'' and ``capture'' actions. Crucially, a successful capture requires simultaneous execution by at least two predators, yielding a positive team reward of 10, whereas a unilateral capture attempt results in prey escape and a hyperparameter-defined negative penalty. Similar to the Matrix Games environment, this setting features highly ``risky'' optimal joint actions: while successful coordination yields the maximum return, any unilateral deviation by a teammate triggers a steep penalty. Consequently, this scenario poses substantial challenges for accurate value estimation within conventional value decomposition frameworks, as the severe penalties for miscoordination can easily drive agents to underestimate the true values of optimal joint actions.

\begin{figure}[t]
    \centering
    \begin{subfigure}[b]{0.25\linewidth}
        \centering
        \includegraphics[width=\linewidth]{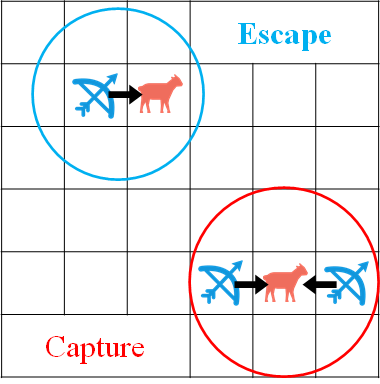}
        \caption{Diagram of Predator and Prey}
        \label{fig:left}
    \end{subfigure}
    \hspace{0.05\linewidth} %
    \begin{subfigure}[b]{0.65\linewidth}
        \centering
        \includegraphics[width=\linewidth]{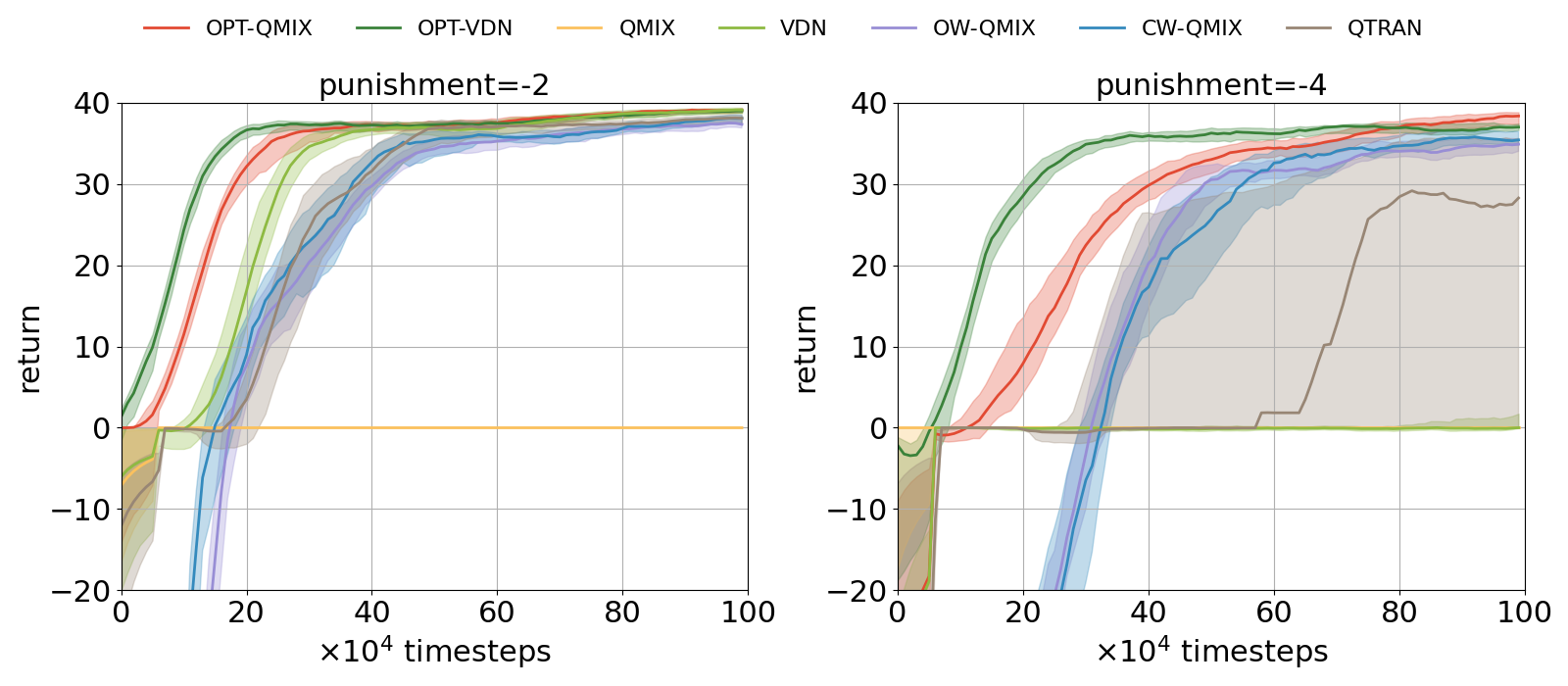}
        \caption{Average Returns of Different Algorithms on Predator and Prey Over Time}
        \label{fig:right}
    \end{subfigure}
    \caption{Experimental Results of Predator and Prey}
    \label{fig:2}
\end{figure}

We evaluate the algorithms under two distinct configurations featuring miscapture penalties of -2 and -4. A penalty of -4 imposes a strictly higher demand on effective coordination, significantly increasing the likelihood of diminished team rewards for misaligned actions. Figure \ref{fig:2}(b) illustrates the average returns achieved by the respective algorithms throughout the training process. As the results indicate, our proposed variants, OPT-VDN and OPT-QMIX, consistently converge to the highest final returns with markedly faster convergence rates across all settings. This superior performance demonstrates that our method reliably overcomes reward ambiguities and successfully steers the joint policy toward optimal joint actions. In contrast, the conventional baseline algorithms encounter severe difficulties with the inconsistent rewards generated by the same action under different circumstances. QMIX struggles to interpret these conflicting signals and consequently converges toward suboptimal joint actions (specifically, a state of inaction), settling for a return of zero. VDN achieves marginal robustness owing to its simple structural design and manages to learn the optimal joint policy under the milder penalty of -2, but it completely fails to identify optimal joint actions when faced with the stricter penalty of -4. This behavior highlights the fundamental limitations of conventional monotonic value decomposition in environments characterized by non-monotonic action-reward relationships. Enhanced algorithms such as OW-QMIX, CW-QMIX, and QTRAN demonstrate the feasibility of overcoming ambiguous rewards. However, this capability comes at the cost of substantially slower convergence rates compared to our exploration strategy. Specifically, OW-QMIX and CW-QMIX suffer from reduced convergence speeds because they apply weighting mechanisms directly to the TD loss of the main action-value network, which can overly restrict the policy optimization process. Meanwhile, QTRAN experiences delayed convergence due to the optimization difficulties and computational overhead introduced by its highly complex network architecture.

\subsection{StarCraft Multi-Agent Challenge}\label{subsection5.4}

We evaluate our proposed variants OPT-VDN and OPT-QMIX in the StarCraft Multi-Agent Challenge (SMAC) environment \cite{samvelyan2019starcraft} to examine their effectiveness in tackling complex scenarios that require precise multi-agent coordination. SMAC serves as a prominent benchmark for MARL, built upon the popular real-time strategy game StarCraft II. It offers a diverse set of team-based combat simulations, where multiple agents must cooperate to eliminate the opposing team and secure victory. In these simulations, the agents of the allied team are controlled by MARL algorithms, while the opposing enemy units operate under the control of the game's built-in heuristic algorithm. The experiments in this subsection are conducted across six distinct maps: ``3s5z'', ``1c3s5z'', ``5m\_vs\_6m'', ``8m\_vs\_9m'', ``10m\_vs\_11m'', and ``MMM2''. These maps cover a wide range of difficulty levels, officially classified as ``Easy'', ``Hard'', and ``Super Hard'' within the SMAC benchmark. Table \ref{table:3} provides detailed information for each map. These scenarios present varying levels of complexity, scaling from basic coordination challenges to highly advanced combats that demand sophisticated joint policies to handle asymmetric team compositions and adversarial dynamics.

\begin{table*}[t]
\centering %
\begin{tabular}{|c|c|c|c|c|} %
\hline
\textbf{Name} & \textbf{Ally Units} & \textbf{Enemy Units} & \textbf{Difficulty}\\ %
\hline
\hline
3s5z & 3 Stalkers \& 5 Zealots & 3 Stalkers \& 5 Zealots & Easy \\ %

1c3s5z & \makecell{1 Colossi \& 3 Stalkers \\ \& 5 Zealots} & \makecell{1 Colossi \& 3 Stalkers\\ \& 5 Zealots} & Easy\\ %

5m\_vs\_6m & 5 Marines & 6 Marines  & Hard\\ %

8m\_vs\_9m & 8 Marines & 9 Marines  & Hard\\ %

10m\_vs\_11m & 10 Marines & 11 Marines  & Hard\\ %

MMM2 & \makecell{1 Medivac \& 2 Marauders\\ \& 7 Marines} & \makecell{1 Medivac \& 3 Marauders \\ \& 8 Marines} & Super Hard\\  %
\hline
\end{tabular}
\caption{Compositions and Difficulty Ratings of SMAC}
\label{table:3}
\end{table*}

\begin{figure*}[t!]
    \centering
    \includegraphics[width=1\textwidth]{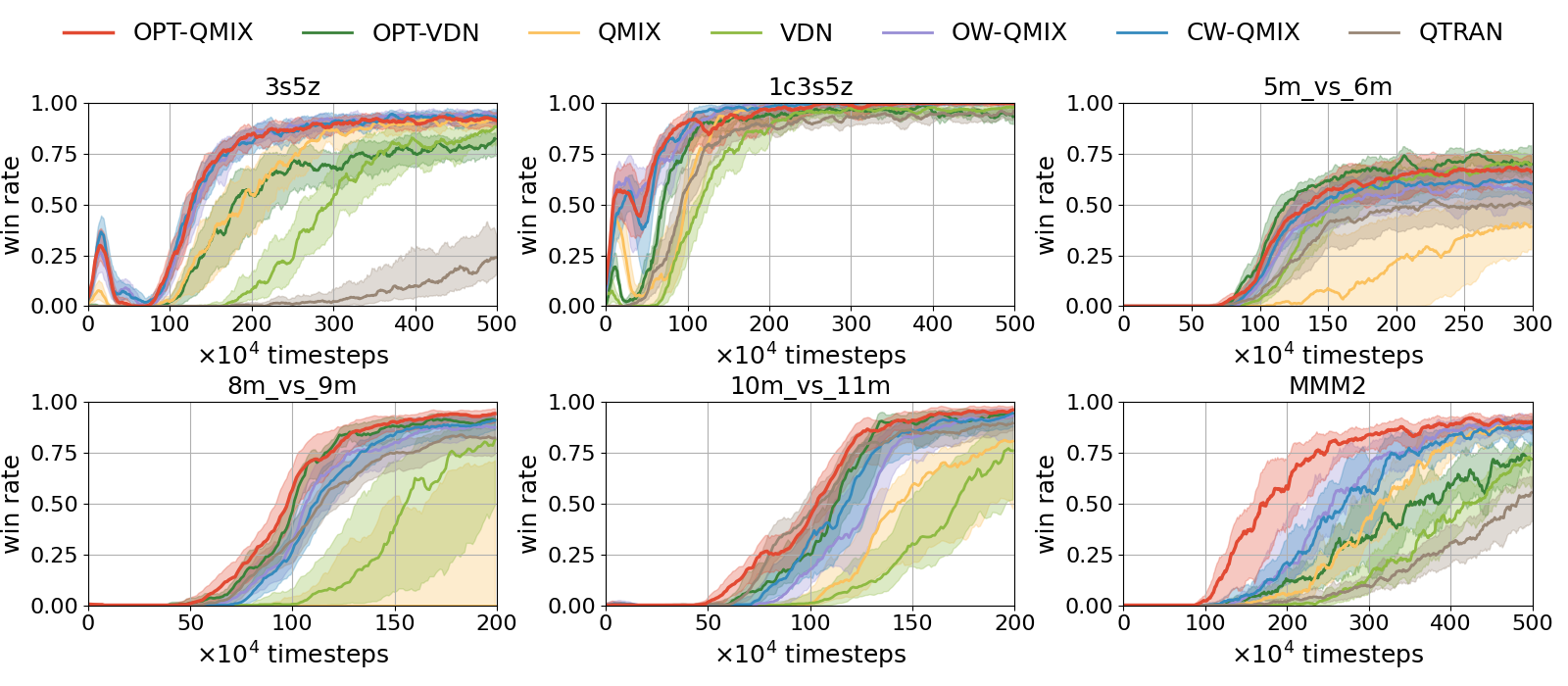}
    \caption{Experimental Results of StarCraft Multi-Agent Challenge}
    \label{fig:3}
\end{figure*}

Figure \ref{fig:3} illustrates the win rates of various algorithms across different maps throughout the training process. The experimental results reveal notable differences in performance among the algorithms. Specifically, VDN and QMIX exhibit limited robustness to extended exploration. This vulnerability primarily arises because their conventionally applied exploration strategies lack the capability to guide the joint policy out of suboptimal joint actions, leaving their monotonic aggregation functions highly susceptible to the value underestimation problem. This limitation becomes particularly evident on maps officially classified as ``Hard'' or ``Super Hard'', where the need for nuanced coordination and adaptive joint policies is significantly greater. In comparison, our proposed variants, OPT-VDN and OPT-QMIX, achieve higher win rates and lower variance across nearly all maps compared to VDN and QMIX, while demonstrating superior robustness to extended exploration. Without altering the core algorithmic framework, OPT-VDN and OPT-QMIX successfully mitigates the value underestimation problem inherent to monotonic aggregation functions by incorporating the Optimistic $\epsilon$-Greedy Exploration strategy. This addition actively drives the joint policy toward optimal joint actions within potentially high-return regions of the policy space. Compared to other enhanced algorithms such as OW-QMIX, CW-QMIX, and QTRAN, our superior variant OPT-QMIX achieves higher or comparable final win rates across all maps. Furthermore, its performance curves are consistently positioned to the upper left of its competitors, demonstrating a distinct advantage in convergence speed. This advantage can be attributed to our strategy's focus on leveraging optimism to guide action selection probabilities during exploration, rather than relying solely on complex architectural enhancements such as QTRAN or manual weighting mechanisms applied directly to the TD loss such as OW/CW-QMIX. These findings provide compelling evidence for the overall effectiveness of the Optimistic $\epsilon$-Greedy Exploration strategy.

\subsection{Ablation Experiments Across Different Exploration Strategies}\label{subsection5.5}

In this subsection, we compare our \textbf{Optimistic $\epsilon$-Greedy Exploration} strategy with other widely used stochastic exploration strategies in MARL to highlight the unique ability of our approach to guide the joint policy toward optimal joint actions and achieve superior performance. To validate this, we use the QMIX framework as a baseline and replace its exploration strategy with our Optimistic $\epsilon$-Greedy Exploration strategy, the conventional $\epsilon$-greedy exploration strategy, a noise-based exploration strategy, and an intrinsic reward-based exploration strategy (curiosity-driven exploration). The conventional $\epsilon$-greedy exploration strategy selects actions uniformly at random with a probability of $\epsilon$, encouraging exploration regardless of their action-value estimations. The noise-based exploration strategy incorporates additional noise into the value estimates for each action and consistently selects the action with the highest noise-augmented value. The intrinsic reward-based exploration strategy introduces an extra reward signal during training, based on the gap between the predicted return and the actual return. This incentivizes agents to explore scenarios where their value estimates are less precise or inaccurate, pushing them toward uncharted areas of the state-action space. We provide the detailed implementation of the above exploration strategies within the QMIX architecture in Appendix \ref{secB}.

\begin{figure*}[t!]
    \centering
    \includegraphics[width=1\textwidth]{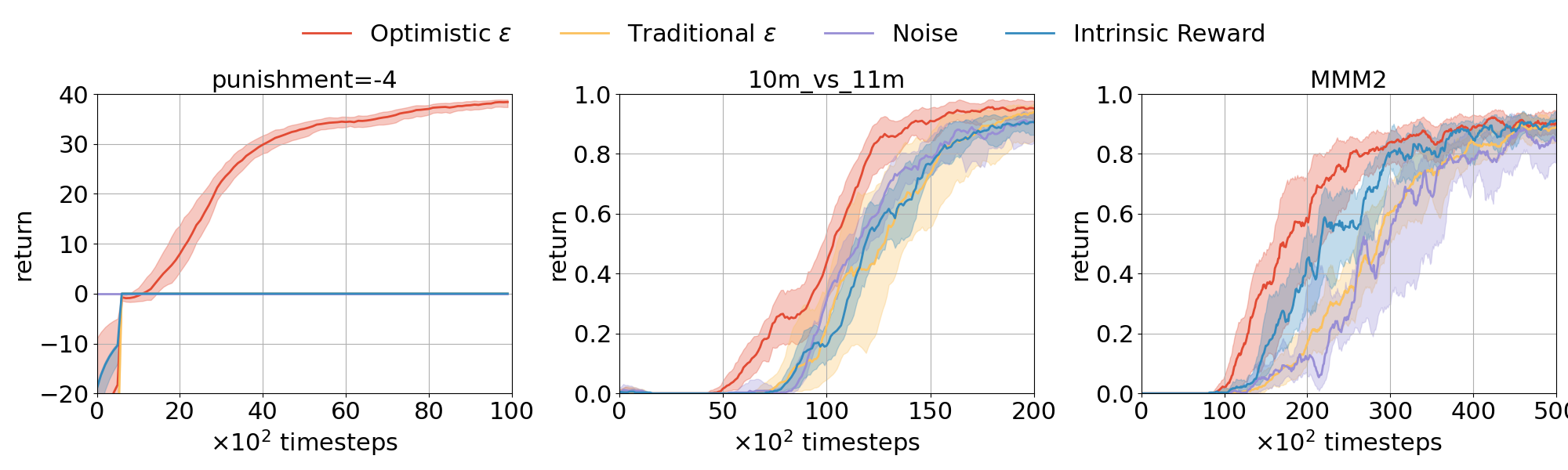}
    \caption{Ablation Experimental Results Across Different Exploration Strategies}
    \label{fig:4}
\end{figure*}

We performed experiments in the Predator and Prey environment, with the miscapture penalty set to -4, as well as in two representative SMAC maps: ``10m\_vs\_11m'' and ``MMM2'', which encompass varying levels of challenge and complexity. Figure \ref{fig:4} presents the average returns and win rates across different exploration strategies in these environments. Remarkably, all stochastic exploration methods, except our Optimistic $\epsilon$-Greedy Exploration strategy, learned suboptimal joint actions in the Predator and Prey environment, essentially resulting in agents ``doing nothing'' and achieving a return of 0. This striking result underscores the unique advantage of our method in effectively steering the joint policy toward optimal joint actions. In the SMAC environments, the conventional $\epsilon$-greedy exploration strategy and the noise-based exploration strategy delivered moderate performance, showing limited capability in navigating the complex state-action dynamics. While the intrinsic reward-based exploration strategy encourages agents to explore states with inaccurate value estimates by providing additional incentives based on the discrepancy between estimated values and actual rewards, this mechanism comes at the cost of increased training instability. Specifically, these supplementary reward signals introduce substantial variance during optimization. In comparison, our Optimistic $\epsilon$-Greedy Exploration strategy demonstrated notable advantages, achieving higher final performance, faster convergence rates, and better training stability. These ablation experiments further support that, compared to other stochastic exploration strategies, the Optimistic $\epsilon$-Greedy Exploration strategy aligns more effectively with monotonic value decomposition paradigms. These findings reinforce the effectiveness and adaptability of our approach in addressing the value underestimation problem and other challenges inherent in MARL tasks.

\begin{figure*}[t!]
    \centering
    \includegraphics[width=0.95\textwidth]{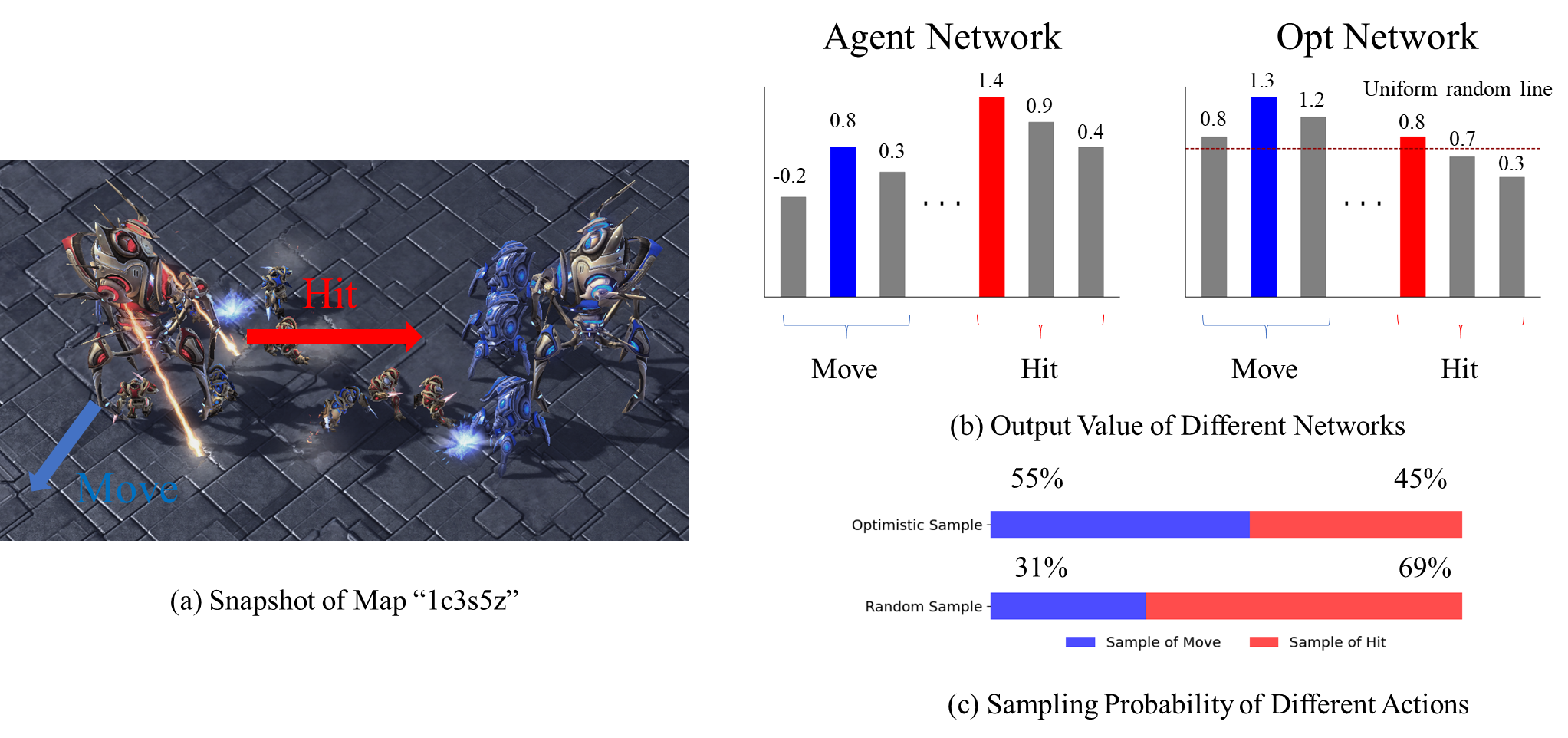}
    \caption{Visualized Explanation of The Optimistic $\epsilon$-Greedy Exploration Strategy}
    \label{fig:5}
\end{figure*}

\subsection{Qualitative Analysis}\label{subsection5.6}

In this subsection, we present a qualitative analysis of our \textbf{Optimistic $\epsilon$-Greedy Exploration} strategy to better illustrate its underlying operational mechanisms. For this analysis, we utilize the SMAC map ``1c3s5z'' as the test scenario. In this environment, the Colossus serves as a crucial allied unit due to its large size, high damage output, and ability to deliver area-of-effect attacks. However, the Colossus's prominence comes with significant vulnerabilities; enemy ranged units can easily focus their fire on it, often leading to its early elimination and placing the allied team at a substantial disadvantage. To mitigate this, an effective policy for the Colossus involves a ``hit-and-run'' strategy: repositioning between attacks to maintain its survivability and maximize its damage contribution over time. Unfortunately, the reward structure of SMAC introduces challenges to learning this tactic. Units receive positive rewards primarily for successfully dealing damage to opponents, while no immediate reward is provided for movement actions. This sparse reward signal often causes conventional baselines to undervalue the ``move'' action, severely limiting the agent's ability to execute efficient positioning strategies.

Figure \ref{fig:5} illustrates a snapshot of our OPT-QMIX method on the ``1c3s5z'' map, trained for 500k timesteps using the parameter settings detailed in Section \ref{subsection5.1}, alongside a subset of the corresponding network output values for the Colossus agent. The \textbf{Agent} network computes the standard expected returns, whereas the \textbf{Opt} network estimates the maximum achievable returns using its monotonically non-decreasing update mechanism. While the \textbf{Agent} network significantly undervalues the ``move'' action due to the aforementioned reward limitations, the \textbf{Opt} network accurately identifies its potential as a critical component of a optimal joint action. By leveraging these optimistic estimates during the $\epsilon$-exploration phase, our strategy significantly increases the action selection probability of the ``move'' action compared to conventional uniform random exploration. Consequently, the Colossus agent is encouraged to attempt the ``move'' action more frequently, ultimately discovering the higher true returns associated with successful hit-and-run coordination. This optimistic sampling loop enables the main Agent network to rapidly overcome its initial underestimation and correctly adjust its valuation of the ``move'' action. Through this mechanism, our approach effectively mitigates the value underestimation problem, facilitating highly directed exploration and successfully guiding the joint policy toward optimal joint actions.

\section{Conclusion}\label{sec:6}

We proposed the Optimistic $\epsilon$-Greedy Exploration strategy for cooperative MARL, which aims to explore optimal joint actions by introducing optimistic action-value networks updated in a monotonically non-decreasing manner. We proved theoretically that these networks converge in probability to the maximum achievable returns. By leveraging these optimistic estimates, our exploration strategy encourages agents to choose optimal joint actions more frequently during the exploration phase. We evaluated the proposed strategy with value decomposition methods across several cooperative multi-agent environments of varying complexity, demonstrating that it effectively guides the joint policy toward optimal joint actions. Furthermore, it achieved higher final returns, significantly superior win rates, and faster convergence compared to other enhanced value decomposition algorithms such as WQMIX, and QTRAN. Crucially, by independently identifying the optimal action components for each individual agent, our exploration strategy maintains a linear computational complexity of $\mathcal{O}(|A|)$ with respect to the action space of each agent, allowing efficient scaling in the number of actions.

While our proposed method demonstrates substantial improvements, it is important to acknowledge its limitations to guide future research. First, our current theoretical analysis and practical implementation are strictly framed within conventional monotonic value decomposition paradigms. Investigating how optimism-guided exploration mechanisms can be adapted for and integrated into broader MARL paradigms, such as multi-agent policy gradient methods, remains an important open question. Second, our approach fundamentally assumes a fully cooperative Dec-POMDP setting. Relaxing this assumption to address wider classes of general-sum games, similar to the extensions explored in ParetoAC, would significantly broaden its applicability. Furthermore, from a training dynamics perspective, optimistic methods show great promise in long-term policy guidance but can occasionally face early-stage instability due to neural network approximation errors. Incorporating prior knowledge, such as task representations, could mitigate this by stabilizing the training process and accelerating early-stage learning.

\begin{appendices}

\section{Proof of Lemmas and Theorems}\label{secA}

\setcounter{theorem}{1}
\setcounter{lemma}{0}

\begin{lemma}
For the estimation function sequence $\{f_t(x)\}$ defined in Eq. \eqref{eq:10}, if initialized with $f_0(x)\leq r_{max}$, where $r_{max}=\max_t r_t$, then $\forall t\geq0$, it holds that $f_t(x)\leq r_{max}$.
\end{lemma}

\begin{proof}

    We prove the lemma by employing the method of mathematical induction.

    When $t=0$, since the sequence is initialized with $f_0(x)\leq r_{max}$, the base case holds.

    Assume that the lemma holds true for some arbitrary positive integer $t$, that is $f_t(x)\leq r_{max}$. We will now show that the lemma holds for $t+1$. $f_{t+1}(x)$ satisfies:

    \begin{equation*}
        \begin{aligned}
            f_{t+1}(x) &= f_t(x) + \alpha( r_t-f_t( x ))\\
            &\leq f_t(x) + r_t-f_t( x )     & (\alpha \in [0,1])\\
            &\leq r_{t}\\
            &\leq r_{max}    & (r_t(a) \leq r_t(a^*))
        \end{aligned}
    \end{equation*}

    Thus, the lemma holds for $t+1$. By the principle of mathematical induction, the lemma holds for all $t \geq 0$.
    
\end{proof}

\begin{lemma}
    Given an auxiliary function sequence updating in the following form:
    \begin{equation*}
	\begin{aligned}
		f_{t+1}^{'}\left( x \right)= \begin{cases}
		      f_t^{'}\left( x \right) +\alpha \left( r_t-f_t^{'}\left( x \right) \right)&		if\,\,r_t=r_{max}\\
				f_t^{'}\left( x \right)&		else\\
		\end{cases}
	\end{aligned}
    \end{equation*}
    For the estimation function sequence $\{f_t(x)\}$ defined in Eq. \eqref{eq:10}, if initialized with $f_0(x) = f_0^{'}(x)$, then $\forall t\geq 0$, it holds that $f_t(x)\geq f_t^{'}(x)$.
\end{lemma}

\begin{proof}
    We prove the lemma by employing the method of mathematical induction.

    When $t=0$, since the sequence is initialized with $f_0(x) = f_0^{'}(x)$, the base case holds.

    Assume that the lemma holds true for some arbitrary positive integer $t$, that is $f_t(x)\geq f_t^{'}(x)$. We will now show that the lemma holds for $t+1$.

    If $r_t\leq f_t^{'}(x)$, $f_{t+1}(x)$ and $f_{t+1}^{'}(x)$ satisfy:

    \begin{equation*}
	\begin{aligned}
		f_{t+1}(x) = f_t(x)\geq f_t^{'}(x) = f_{t+1}^{'}(x)
	\end{aligned}
    \end{equation*}

    If $f_t^{'}(x) < r_t \leq f_t(x)$, $f_{t+1}(x)$ and $f_{t+1}^{'}(x)$ satisfy:

    \begin{equation*}
	\begin{aligned}
		f_{t+1}(x) &= f_t(x)\\
		&\geq r_t\\
		&= r_t + f_t^{'}(x) -f_t^{'}(x)\\
		&\geq f_t^{'}(x) + \alpha(r_t - f_t^{'}(x))    & (\alpha \in [0,1])\\
		&=f_{t+1}^{'}(x)
	\end{aligned}
    \end{equation*}

    If $r_t > f_t(x)$, $f_{t+1}(x)$ and $f_{t+1}^{'}(x)$ satisfy:

    \begin{equation*}
	\begin{aligned}
		f_{t+1}(x) &= f_t(x)+\alpha(r_t - f_t(x))\\
		&=\alpha r_t + (1-\alpha)f_t(x)\\
		&\geq \alpha r_t + (1-\alpha)f_t^{'}(x)\\
		&=f_{t+1}^{'}(x)
	\end{aligned}
    \end{equation*}

    Therefore, regardless of the value of $r_t$, it always holds that $f_{t+1}(x)\geq f_{t+1}^{'}(x)$. By the principle of mathematical induction, the lemma holds for all $t \geq 0$.
    
\end{proof}
\begin{theorem}
The Optimistic $\epsilon$-Greedy Exploration strategy defined in Eq. \eqref{eq:12} samples the optimal action of each agent with a higher probability than the traditional $\epsilon$-greedy exploration strategy defined in Eq. \eqref{eq:3}.
\end{theorem}

\begin{proof}
    Considering that temperature normalization does not affect the relative magnitude of values, we analyze the influence of $f_i(x)$ in this proof. We first prove $Softmax(f(a^*))\geq\frac{1}{|a|}$ by contradiction. Assume, for the sake of contradiction, that $Softmax(f(a^*))<\frac{1}{|a|}$. For all $a\neq a^*$, it holds that $f(a)\leq f(a^*)$; therefore, we have:

    \begin{equation*}
	\begin{aligned}
		\sum_{a\in A}{Softmax(f(a))} &\leq \sum_{a\in A}{Softmax(f(a^*))}\\
        &= |a|Softmax(f(a^*))\\
        &< |a|\times \frac{1}{|a|}\\
        &= 1
	\end{aligned}
    \end{equation*}

    Now, we have shown that $\sum_{a\in A}{Softmax(f(a))}<1$. However, this contradicts the definition of the Softmax function, which satisfies $\sum{Softmax(\cdot)}=1$. Since assuming that $Softmax(f(a^*))<\frac{1}{|a|}$ leads to a contradiction, we conclude that $Softmax(f(a^*))\geq\frac{1}{|a|}$.

    Building upon this, we prove this theorem by considering all possible cases for the optimal action component $a^*$.

    \noindent \textbf{Case 1:} $a^* = arg\max Q(\tau,a)$

    The \textbf{Optimistic $\epsilon$-Greedy Exploration} strategy selects the optimal action with a probability of $1-\epsilon+\epsilon Softmax(a^*)$, whereas the conventional $\epsilon$-greedy exploration strategy selects it with a probability of $1-\epsilon+\frac{\epsilon}{|a|}$. Since it has already been proven that $Softmax(f(a^*))\geq\frac{1}{|a|}$, this demonstrates that the \textbf{Optimistic $\epsilon$-Greedy Exploration} strategy provides a higher probability of selecting the optimal action.

    \noindent \textbf{Case 2:} $a^* \neq arg\max Q(\tau,a)$

    The Optimistic $\epsilon$-Greedy Exploration strategy selects the optimal action with a probability of $\epsilon Softmax(a^*)$, whereas the conventional $\epsilon$-greedy exploration strategy selects it with a probability of $\frac{\epsilon}{|a|}$. Since it has already been proven that $Softmax(f(a^*))\geq\frac{1}{|a|}$, this demonstrates that the Optimistic $\epsilon$-Greedy Exploration strategy provides a higher probability of selecting the optimal action.

    In both cases, we have shown that the Optimistic $\epsilon$-Greedy Exploration strategy provides a higher probability of selecting the optimal action. Therefore, the theorem holds true for all possible situations regarding the optimal action $a^*$.
\end{proof}

\section{Implementation of Different Exploration Strategies}\label{secB}

During our ablation experiments, we incorporated several stochastic exploration strategies into the QMIX framework, including the conventional $\epsilon$-greedy exploration, noise-based exploration, and intrinsic reward-based exploration. In this subsection, we detail the specific implementations of each strategy within QMIX to ensure experimental fairness and reproducibility.

\textbf{Conventional $\epsilon$-Greedy Exploration:} We implement the conventional $\epsilon$-greedy exploration strategy based on the reinforcement learning textbook \cite{sutton1998reinforcement}. Specifically, at each environment step, the agent generates a random variable. With probability $\epsilon$, the agent explores by uniformly sampling an action from the set of currently available actions. Conversely, with probability $1-\epsilon$, the agent exploits by selecting the action that yields the maximum estimated action-value (Q-value). To handle dynamic action spaces within the environment, an action masking mechanism is employed, which assigns a value of negative infinity ($-\infty$) to unavailable actions to ensure they are strictly excluded from selection. Furthermore, the exploration rate $\epsilon$ follows an annealing schedule: it is initialized to a higher value to encourage broad exploration at the beginning of training, and decays to a predefined lower bound. It is worth noting that during the evaluation phase, $\epsilon$ is strictly set to $0.0$ to ensure purely greedy action selection based on the learned policy. During our experiments, the annealing time and annealing function are set according to the configurations detailed in Section \ref{subsection5.1}.

\textbf{Noise-Based Exploration:} We implement the noise-based exploration following the Noise-Regularized Advantage Value method \cite{wang2022noise}. In contrast to the conventional $\epsilon$-greedy method, this strategy introduces stochasticity by directly injecting scaled uniform noise into the estimated Q-values prior to action selection. Specifically, at each time step, the algorithm first calculates the numerical range of the current Q-values (i.e., the difference between the maximum and minimum Q-values across all actions). It then generates uniformly distributed random noise within the interval $[-1, 1]$. This noise is proportionally scaled by both the calculated Q-value range and the current exploration rate $\epsilon$, and is subsequently added to the original Q-values to produce perturbed value estimates. To handle dynamic action spaces, the action masking mechanism is applied, overriding the perturbed Q-values of unavailable actions to negative infinity ($-\infty$). Ultimately, the agent selects the action that maximizes this noisy Q-value. Similar to the conventional approach, $\epsilon$ acts as a noise scaling factor and follows an annealing schedule, which gradually reduces the magnitude of the injected noise over time. During the evaluation phase, $\epsilon$ is strictly set to $0.0$, eliminating the noise entirely for purely greedy execution. To ensure experimental fairness through hyperparameter consistency, the annealing duration and the annealing function of the noise scaling factor $\epsilon$ are kept identical to those in the conventional $\epsilon$-greedy exploration strategy.

\textbf{Intrinsic Reward-Based Exploration:} We implement the intrinsic reward-based exploration following DDN \cite{zhou2025double}. Built upon the conventional $\epsilon$-greedy method, this strategy introduces an additional reward signal to encourage the agent to explore states where the reward estimates are not yet accurate. Following the implementation of DDN, we approximate the true environmental expected return using the output of the unconstrained joint action-value network $Q_{jt}$, and use the discrepancy between this target and the global action-value estimation $Q_{tot}$ as the basis for quantifying the agent's intrinsic reward. We use the product of the absolute value of this discrepancy and an intrinsic reward scaling factor as the supplementary reward signal to encourage the agent to explore states with inaccurate value estimations during the exploration process. In this experiment, the scaling factor is set to a constant of 0.1. Aside from introducing the additional reward, the action selection mechanism of this method is identical to the conventional $\epsilon$-greedy approach. In our experiments, we adopt the same parameter settings as the conventional $\epsilon$-greedy method described previously.

\end{appendices}

\bibliography{sn-bibliography}%

\end{document}